\documentclass{llncs}
\usepackage[T1]{fontenc}    % Encodage des accents
\usepackage{lmodern}
\usepackage[utf8]{inputenc}
\usepackage{amsmath}        % La base pour les maths
\usepackage{amsfonts}    
\usepackage{anyfontsize}
\usepackage[svgnames]{xcolor} % De la couleur

\usepackage{graphicx} % inclusion des graphiques
\usepackage{wrapfig}  % Dessins dans le texte.
\usepackage{amssymb, bm}
\usepackage{xparse}
\usepackage{csquotes}
\usepackage{stmaryrd}
\usepackage{enumitem}
\usepackage[ruled,vlined,linesnumbered]{algorithm2e}

\SetAlgoSkip{}

\usepackage[title]{appendix}
\usepackage{subfiles} % Best loaded last in the preamble

\renewcommand{\underline}[1]{\textbf{#1}}
\renewcommand{\L}{\ensuremath{\mathcal{L}}}
\newcommand{\A}{\ensuremath{\mathcal{A}}}
\newcommand{\N}{\ensuremath{\mathbb{N}}}
\newcommand{\Z}{\ensuremath{\mathbb{Z}}}
\newcommand{\R}{\ensuremath{\mathbb{R}}}
\newcommand{\D}{\ensuremath{\mathbb{D}}}
\renewcommand{\P}{\ensuremath{\mathbb{P}}}
\newcommand{\tS}{\ensuremath{\mid\!S\!\mid}}
\newcommand{\sem}[1]{\llbracket#1\rrbracket}
\newcommand{\inter}[2]{[#1,#2[}

\newcommand{\nomdelalgo}{SAI\xspace}
\newcommand{\etal}{\emph{et al.}}

\title{Passive Learning of Symbolic Automata over Monotonic Algebras}
\author{Peter Habermehl\inst{1}\orcidID{0000-0002-7982-0946} and Erwann Loulergue\inst{2}\orcidID{0009-0006-4965-2634}}

\institute{Université Paris Cité, CNRS, IRIF, F-75013, Paris, France\\
  \email{haberm@irif.fr}
  \and {Université Paris-Saclay, CNRS, ENS Paris-Saclay, Laboratoire Méthodes Formelles, 91190, Gif-sur-Yvette, France\\
    \email{ErwannL@lmf.cnrs.fr}}
  }

\begin{document}
    \maketitle
	\begin{abstract}
  Symbolic automata extend classical finite-state automata to handle large or infinite alphabets by
labeling transitions by predicates coming from
  a boolean algebra.
  Many results from automata theory have been lifted to this model, and it has proved its usefulness for example in multiple software verification applications. Here, 
  we tackle the passive learning problem of identification in the limit, i.e. learning a model from a sample without access to an oracle to query.  
  We provide an algorithm, \nomdelalgo, that efficiently identifies in the limit symbolic automata over any monotonic algebra where predicates labeling transitions are of the form $a\preceq x \prec b$.
  The algorithm extends the RPNI framework for passive learning of finite-state
  automata to symbolic automata thanks to a new \emph{splitting} operation
  inspired by RTI, a passive learning algorithm for
  deterministic real-time automata, a subclass of timed automata.
  The learning algorithm combines merging of states
  and splitting of states allowing to infer the predicates on transitions
  in a top-down fashion.
  We prove that \nomdelalgo admits polynomial size characteristic samples.

\end{abstract}

    \section{Introduction}\label{introduction}
        Symbolic finite-state automata (SFA) have attracted a lot of interest recently.
In contrast to usual finite-state automata (FA) they can define languages
over an infinite alphabet. Their transitions are %not letters from a finite alphabet but
predicates coming from a boolean algebra denoting
sets of letters from a potentially infinite domain
(natural numbers, reals, etc.). Even if the alphabet is finite, it is often
useful to define predicates describing succinctly finite sets of letters
like all characters of some kind in Unicode.
There has been a flurry of papers lifting results from the
setting of FA to SFA (overviewed in \cite{DAntoniV21}).
For example, determinization, minimization, $\omega$-automata,
etc. have been considered.
SFA are used successfully in practice \cite{DAntoniV21}.

The problem of exact \emph{learning} FA and extensions is fundamental \cite{Gold67,Gold78,Angluin87}. There are
roughly two types of learning. Passive learning considers inferring an automaton
from a sample of (positive and negative) words from
an unknown \emph{target} automaton.
The learned automaton should be consistent with the sample,
i.e. accept \emph{all} positive examples and reject \emph{all} negative ones. 
Furthermore, it should %in some way
\emph{generalize}, i.e. accept more words
than the positive examples.
One would like to have a theoretical guarantee that the learning algorithm
\emph{identifies in the limit} the target automaton i.e.
given more and more examples it converges to the target automaton
at some point. % (when this happens is usually unknown).
A \emph{characteristic sample} (CS) is
a sample which guarantees that the algorithm infers the target automaton
given a sample containing at least the CS.
A class of languages (or automata representing them) is called
\emph{efficiently identifiable}, if such a (polynomial)
learning algorithm exists, with
a CS of polynomial size which can be computed efficiently.
FA are efficiently identifiable \cite{dlH10} 
and a lot of passive learning algorithms have been proposed \cite{RPNI,dlH10}.
These algorithms do not necessarily have CS but work well in practice.
A lot of algorithms are based on starting with a
\emph{prefix-tree automaton} for the sample which only accepts 
the positive words in the sample and then \emph{merging} states
in some fixed or heuristically chosen order generalizing
the language accepted by the automaton while ensuring that negative
examples are all rejected.

In active learning pioneered by Angluin \cite{Angluin87}, a learner
tries to infer a target automaton
by asking membership and equivalence queries to a teacher which knows the target.
Angluin's L$^*$ algorithm for FA has been extended
to SFA \cite{MATstar} recently.  %after several preliminary works \cite{}.
The main difficulty is that the learner 
not only must infer the structure of the automaton but also the predicates
labeling transitions. This is solved elegantly in \cite{MATstar} by
using an active learning algorithm for predicates of the boolean algebra
as subroutines in automata learning.
Coming back to passive learning of SFA considered in this paper,
here the difficulty is the same: How to combine inference of the structure
of the automaton with inference of predicates of the boolean algebra?
Passive learning for SFA has been considered very recently in~\cite{FFZ}.
Among others, they give a sufficient condition for efficient learnability of
SFA which allows one to show that SFA over any \emph{monotonic} algebra (predicates
can be seen as intervals) are efficiently learnable.
Their learning algorithm is based on the existence of a
\emph{generalization} function over the alphabet which given a set of
domain values outputs predicates in the underlying
boolean algebra generalizing them.
Then, passive learning of an SFA from a sample $S$
is done by first \emph{decontaminating} $S$,
then passively learning a DFA $A_S$ over
the concrete alphabet and then finally
generalizing $A_S$ to get an SFA.
However, the obtained SFA might not be consistent with
the sample. In that case, the algorithm just outputs the prefix-tree automaton.
, that is, 
the simplest automaton that accepts only the positively-labeled words of the sample.
There is a CS of polynomial size whose presence in the sample
guarantees identification, but it can be said that the algorithm
is only of theoretical interest as in practice there is no
guarantee that the sample contains a CS and one would like to
obtain a generalization of the sample regardless.

In a more practically oriented paper
\cite{VWW} passive learning of deterministic real-time automata (DRTA)
is considered.
DRTA are a restricted form of timed-automata where only constraints
on the time spent between two consecutive events can be expressed.
Therefore transitions are labelled with (time) intervals.
Therefore, DRTA can be seen as a special case of SFA.
The problem of inferring predicates on the transitions is elegantly solved
by starting with a prefix-tree automaton which is typically
contradictory (it contains states which should be accepting and rejecting at
the same time) with all transitions labeled by $\top$ denoting the
whole underlying domain.
Then, an operation allows to split states to 
solve some of the contradictions.
This is mixed with merging as in classical passive learning algorithms.
The order splitting and merging is done guarantees that the algorithm
will stop with typically a generalization of the sample.

In this work, we present a novel passive learning algorithm for SFA
called \nomdelalgo (Symbolic Automaton Inference)
which integrates the learning of the structure of the automaton with
the learning of the predicates. The structure of the automaton is learned
similarly to algorithms for FA, predicates are learned with a \emph{top-down}
approach by starting from an automaton where transitions accept all letters at first and splitting predicates only when necessary.
To tackle the main issue in passively learning SFA which is generalizing the behavior of the model to letters not appearing in the sample, our approach starts
from the most generic model (although contradictory) and then specifies just enough for it to become coherent.
\nomdelalgo is based on adapting the algorithm of \cite{VWW} to 
SFA. We show that our algorithm efficiently
identifies SFA over monotonic algebras.
We exhibit a CS of quadratic size in the size of
the automaton (two exponents smaller than the one from \cite{FFZ}).
%; see Section \ref{Conclusionsect} for details).
Contrary to the algorithm of \cite{FFZ} \nomdelalgo generalizes in all
situations (for infinite alphabets).
The CS is defined directly from the symbolic automaton
whereas \cite{VWW} define the CS operationally on a run of the
algorithm.

The rest of the paper is organized as follows.
In Section \ref{prelimsect} we give preliminary definitions. 
Section \ref{algosect} presents our learning algorithm and analyzes it.
We compare with related work in Section \ref{comparison} before concluding in Section \ref{futurework}.

	\section{Preliminaries}\label{prelimsect}

		\subsection{Symbolic Automata}\label{sfasect}
			Symbolic automata theory (see \cite{SFAVeanes} for a tutorial) extends classical automata theory to encompass more complex, or larger, alphabet structures. For example, a complete automaton to recognize some regular expression over a text encoded in UTF16 will need $2^{16}$ transitions from each state, which quickly becomes far too costly. Instead of using a concretely represented finite alphabet, symbolic automata employ an alphabet implicitly characterized by a Boolean algebra: this allows for handling of massive finite size %(like UTF16 with its $2^{16}$ letters)
or even infinite (\N,\Z,\dots) alphabets.
\begin{definition}
	A \underline{Boolean algebra} is a tuple $(\D,\P,\sem{\cdot},\bot,\top,\wedge,\vee,\neg)$ where $\D$ is a set called \underline{domain}, $\P$ a set of \underline{predicates} including $\top$, $\bot$, closed under boolean operations. $\sem{\cdot}: \P \to 2^{\D}$ is a \underline{semantics function} verifying: $\sem{\bot} = \emptyset$, $\sem{\top} = \D$ and $\forall \varphi,\psi \in \P$: $\sem{\varphi \wedge \psi} = \sem{\varphi} \cap \sem{\psi}$, $\sem{\varphi \vee \psi} = \sem{\varphi} \cup \sem{\psi}$, and $\sem{\neg\varphi} = \D\setminus\sem{\varphi}$.
\end{definition}
A Boolean algebra is said to be \textit{effective} when its boolean operations are computable, and predicate satisfiability is decidable.
All Boolean algebra we will consider in this document are assumed to be effective.
There are two examples of Boolean algebras we will come back to:

\textbf{The interval algebra}, whose domain is $\N$, where predicates are intervals $[a,b[$ $(a\in\N, b\in\N\cup\{+\infty\})$ and their boolean combinations with natural semantics. Similar algebras can be obtained when using $\Z$ or $\R$ instead of $\N$.

\textbf{The propositional algebra} is defined relative to a set of atomic propositions $AP = \left\{p_1,p_2,\dots,p_n\right\}$: in contains these,$\top$,$\bot$, and their boolean closure. Its domain is $\left\{0,1\right\}^{n}$, with semantics given by $\sem{p_i} = \left\{v \in \left\{0,1\right\}^{n}, v[i] = 1\right\}$. This algebra and SFA defined over it play a central role in \textit{model checking} \cite{modelchecking}.

One way of defining a Boolean algebra is by taking a set of atomic formulae, $\P_{0}$, including $\top$ and $\bot$ and setting $\P$ as its closure under boolean operations. In the interval algebra, $\P_{0}$ is the set of intervals, and in the propositional algebra, it is the set of atomic propositions and their negations.

A Boolean algebra is said to be \underline{monotonic}\footnote{This definition slightly differs from the one in \cite{FFZ} for technical reasons related to the symbols $d_{-\infty}$ and $d_{\infty}$.} when 
(1) there exists a total order $\prec$ on $\D$, (2)
there exists symbols $d_{-\infty}\in\D$ and $d_{\infty}$ (not in \D) such that for all $d \in \D$, $ d_{-\infty} \preccurlyeq d \prec d_{\infty}$,
(3) for all atomic predicates $\varphi \in \P_{0}$ there exists $a \in \D, b \in \D\cup\{d_{\infty}\}$ such that $\sem{\varphi} = \{d \in \D, a\preccurlyeq d \prec b\}$.
In a monotonic algebra, atomic predicates will simply be denoted as $\inter{a}{b}$. The interval algebra is an example of a monotonic algebra. These algebras also appear in natural settings (think for example regular expressions \verb|[a-z]|,\verb|[0-9]|, or other ranges of Unicode characters).

\begin{definition}
	A \underline{symbolic automaton} (SFA) is a tuple $\A = (\mathbb{A},Q,q_{\iota},F,\Delta)$ with $\mathbb{A}$ a Boolean algebra, $Q$ a finite set of states, $q_{\iota} \in Q$ the initial state, $F \subseteq Q$ the set of final states, and $\Delta \subseteq Q \times \mathbb{P_{\mathbb{A}}} \times Q$ a finite set of transitions, where $\mathbb{P_{\mathbb{A}}}$ is the set of predicates of $\mathbb{A}$. In a transition $(q_{1},\varphi,q_{2})$, $\varphi$ is called the \underline{guard}.
\end{definition}

\begin{wrapfigure}[8]{r}{0.28\textwidth}
	\centering
	\vspace{-\intextsep}
	\includegraphics[width=0.23\textwidth]{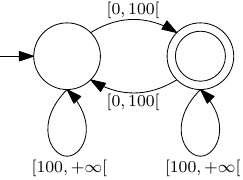}
	\caption{An example SFA}\label{SFAexample}
\end{wrapfigure}

We call letters the elements of $\mathbb{D}$, and words those of $\mathbb{D}^{*}$. An execution of $\A$ on a word $a_1 a_2 \dots a_n$ is a sequence of states $q_0 q_1 \dots q_n$ with $q_0 = q_\iota$ and, for $0\leq i < n$, $a_i \in \sem{\varphi_i}$ for a certain $\varphi_i$ such that $(q_i,\varphi_i,q_{i+1}) \in \Delta$. As usual, an execution is accepting when the last state of the sequence is in $F$, and the language $\L(\A)$ of an automaton is the set of words that label an accepting execution of $\A$.
A SFA is \underline{complete} when at each state $q$, guards of outgoing transitions from $q$ cover the whole of \D: that is, at each state, for each concrete letter there exists at least a transition. A SFA is \underline{normalized} when, given states $p$ and $q$ there is at most one transition from $p$ to $q$ (there might be one from $q$ to $p$). It is \underline{neat} when every transition's guard is a conjunction of atomic predicates. Neat automata value simple transitions at the cost of having a lot of them, while normalized automata do the inverse.
Figure \ref{SFAexample} shows an SFA over the interval algebra which is complete, neat, normalized, and recognizes the language of words over \N\ with an odd number of letters below 100.

	 	\subsection{Automata Learning}\label{learningsect}
	 		Automata learning aims to construct a model, in the form of an automaton, that explains observed behavior (either sampled or queried). In essence, we aim to learn a class of languages: regular, context-free, \ldots over an alphabet $\Sigma$, from examples, i.e. words over $\Sigma$ with a label indicating whether they are member of the target language or not. In a theoretical setting, the examples come from an oracle, also called teacher. The algorithm that formulates hypotheses is often called the learner. In practice, the sample might come from observations, or a black-box software or hardware to which we provide inputs and observe outputs. 

Multiple paradigms exist, some allowing errors and trying to minimize it and some disallowing errors entirely: this is called \emph{exact} learning, which we will focus on. In exact learning, there are two main approaches:
In \textit{active learning}, the learner can query an oracle.
Allowed queries depend on the paradigm. They can be for example
membership and equivalence queries like in Angluin's seminal work \cite{Angluin87}.
In \textit{passive learning}, the learner has no control over the data it receives. It might receive it in full as an input, or in a streaming setting. In this paper, we focus on the former case: the algorithms we discuss receive a \underline{sample} as input, which consists of words together with a label indicating whether they are a positive or negative example (that is, a subset of $\Sigma^*\times\{+,-\})$\footnote{This can also take the form of a sample $S_+$ and a sample $S_-$. We will use both conventions interchangeably.
A sample is not contradictory, i.e. $S_+ \cap S_- = \emptyset$.
}, and output one automaton. Examples of passive learning algorithms include Gold's algorithm for DFA \cite{Gold78} and the RPNI framework \cite{RPNI}. 
Measuring the success of an algorithm in this setting is not trivial: only asking that the output automaton correctly accepts all words labeled as positive examples and rejects negative ones is not enough as the language recognized might agree on the sample and still be different from the target language. For example, returning a basic automaton that accepts exactly the positive examples of
the sample (nothing more, nothing less) satisfies this condition but gives no \emph{generalization} of the sample.
One may therefore ask, given a language $\L$, what a sample should contain in order for a passive learner to identify $\L$. An answer to this question was proposed by Gold in the form of \textit{identification in the limit}, and a more general definition was given by de la Higuera \cite{dlHigueraCS}, which we follow. The idea of a \emph{characteristic sample} is that it should describe the target automaton, so that it is not only correctly identified by chance, but that additional correct examples do not confuse the learner (assuming, of course, they are correctly labeled).
%When an algorithm $A$ identifies the class of all automata in the limit, it is said to be complete: this means that for every language, there exists a sample that would make $A$ return an automaton recognizing this language. It is thus complete in the sense that it does not return only automata of a certain shape, or languages with certain properties.

\begin{definition}
	A \underline{characteristic sample} $CS$ for a target language $\L$ and a learning algorithm $A$ is an input sample $\{S_+ \subset \L, S_- \subset \D^*\setminus\L\}$ such that:
	\begin{itemize}[nolistsep]
		\item Given $CS$ as input, $A$ returns an automaton $\A$ with $\L(\A) = \L$,
		\item Given a sample $S' \supset CS$ labeled in accordance with \L, $A$ still returns an automaton with language $\L$.
	\end{itemize}
\end{definition}

\begin{definition}
	A class of automata $\mathcal{C}$ is said to be \underline{identifiable in the limit} when there exists an algorithm $A$ such that for every automaton $\A \in \mathcal{C}$, there exists a characteristic sample of $\L(\A)$.
	If these characteristic samples are of size polynomial in the size of the minimal automaton\footnote{This definition thus depends on the existence of a minimal representant in the class.} recognizing $\L(\A)$, computable in time polynomial in the size of this automaton, and the learning algorithm runs in polynomial time, the class is \underline{efficiently identified}.
\end{definition}

	\section{Learning SFAs}\label{algosect}
		In \cite{FFZ}, the authors have shown that SFAs over a monotonic algebra are efficiently identifiable in the limit, but SFAs over the propositional algebra are not efficiently identifiable unless $P=NP$. %The former fact is a positive result: still, the algorithm has cases where it only returns the \emph{prefix-tree automaton} of the given sample, a simple automaton that accepts only the positively-labeled words of the sample (essentially learning nothing). The latter shows there can be no generic algorithm for the non-monotonic case that would run in polynomial time and admit polynomial-size characteristic samples.
Learning a symbolic automaton from sampled data runs into the problem of deciding how to identify letters that are not in the sample. For example, say the sample contains 0 as a positive example and 100 as a negative one. The automaton learned could accept letters from 0 to 99, or reject all letters from 1 to 100, or accept from 0 to 49 and reject above 50, and still be consistent with the sample. Note that this is inherent to this precise situation: in active learning, further queries could resolve this issue, and in learning DFAs, the alphabet is finite (and typically small) and no line has to be drawn as usually all letters appear in the sample.

\smallskip
We define an algorithm (see Algorithm \ref{FullAlgo}), called \nomdelalgo for Symbolic Automaton Inference, to learn symbolic automata over any \textit{monotonic} algebra, inspired by RTI \cite{VWW}. RTI was defined for a subclass of timed automata, namely \textit{real-time automata}, that simultaneously recognize both a word and a sequence of numbers representing \textit{timings} of the letters from that word. In our case, there are no letters, only timings, if we consider our intervals as windows in time.

The algorithm uses a state-merging framework, also used for DFA learning: it starts with a basic automaton, the \emph{prefix-tree automaton}, then tries to merge the states. More precisely, it uses the red-blue framework \cite{EDSMintro}, where red states are states we have finished identifying, and blue states
(successors of red states) are candidates for becoming red (that might still need more work). During a run, there will be states that are neither blue nor red, and the algorithm will return when all remaining states are red.  As our sample contains both positive and negative examples, some states will be explicitly rejecting (instead of simply not accepting). An automaton with explicitly rejecting states is called \emph{augmented}, and all automata further on will be augmented with a set $R$ of rejecting states. For brevity, coloring of states is stored implicitly in the following algorithms.

At first, only the initial state is red and all guards are $\top$, so the only thing separating examples is their length.
Then, at each iteration, \nomdelalgo tries to, in this order:
\begin{enumerate}[nolistsep]
	\item merge a blue and a red state,
	\item promote a blue state to red,
	\item split a transition from a red to a blue state.
\end{enumerate}
We check that no irreversible inconsistency (see below)
is induced by such an operation : e.g., a state both accepting and rejecting cannot be colored red, and must be split into different states.

\subsection{The Symbolic PTA}\label{SPTAsect}
Usually, in state-merging algorithms, the first automaton constructed is the \emph{prefix-tree automaton}, or PTA, a simple tree-shaped automaton that accepts (all and only) the positive samples. States of this PTA then get merged and colored, constructing the automaton in a \emph{bottom-up} fashion: first all letters are separated and transitions are very specific, then merges allow for generalization. However, if we were to do the same here, all letters being separated would mean no generalization, which defeats the purpose of symbolic automata on huge or infinite alphabets. The introduction of splits allows us to use a \emph{top-down} strategy. 

\begin{wrapfigure}[9]{r}{0.4\textwidth}
	\centering
	%\vspace{-0.6\intextsep}
	\includegraphics[width=0.38\textwidth]{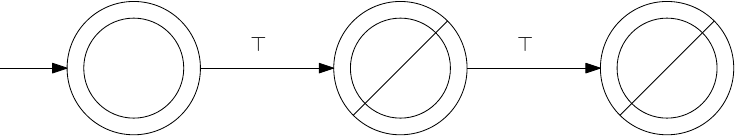}
	\caption{The SPTA for a sample with accepting and rejecting samples of length 1 and 2, and containing $\varepsilon$ as a positive example.}
	\label{SPTAfig}
\end{wrapfigure}
Thus, to create the initial automaton from a sample, called the symbolic PTA (SPTA), all guards are set to $\top$ (i.e. $[d_{-\infty},d_{\infty}[$) and the words from the sample are only distinguished by their lengths. Concretely, the \texttt{SPTA} function takes a sample as input, and outputs an augmented SFA $(\A,R)$ with as many states as the length of the longest word in the sample. The $n$-th state is accepting (resp. rejecting) when there is an $n$-letter word with positive (resp. negative) label in the sample (Figure \ref{SPTAfig}).
Intuitively, this amounts to constructing the usual-sense PTA and then
merging all states with the same depth in the tree.

\subsection{Splitting}\label{split}
\emph{Splitting} a transition $\delta = (q,\inter{a}{b},q')$ at value $c \in \inter{a}{b}$ means creating two transitions from $q$ with guards $\inter{a}{c}$ and $\inter{c}{b}$. This is the operation \cite{VWW} brought to state-merging algorithms: in the usual \emph{bottom-up} strategy, there is no need for it.
For a split operation, we will remove the targeted state and its descendants from $\A$, then split the transition $\delta$ in two, classify samples that fired the removed transition in two sets, and generate the SPTA for these sets.
Every string that fires $\delta$, can be written as $\omega = \tau\cdot t\cdot\tau'$ with $a\preceq t \prec b$. We call  $t\cdot\tau'$ the tail of $\omega$ for $\delta$, and denote it $\omega^\delta$. Every transition being split in \nomdelalgo is leading to an SPTA, whose tree structure ensures the uniqueness of this decomposition (which is not guaranteed in the general case). Denote the subsample of tails from $S$ for $\delta$ as $S^\delta$, that is $S^\delta= (\{\omega^\delta\mid \omega\in S^+ \text{ fires } \delta\},\{\omega^\delta\mid \omega\in S^-\text{ fires } \delta\})$. For a split at value $c$, denote $S^{\delta_1}$ the subset of $S^{\delta}$ with the first letter  $\prec\! c$, and $S^{\delta_2}$ those with first letter $\succeq\! c$.

Using these definitions, the splitting algorithm is given in Algorithm \ref{SplitAlgo}.

\begin{algorithm}[h]
	\DontPrintSemicolon
	\SetKwData{True}{true}\SetKwData{False}{false}
	\SetKwFunction{SPTA}{SPTA}\SetKwFunction{Merge}{Merge}
	\KwData{An augmented SFA $(\A = (\mathbb{A},Q,q_{\iota},F,\Delta), R)$, a transition $\delta =(q,\inter{a}{b},q')$ going to a SPTA, a value $c \in \inter{a}{b}$ , a sample $S$}
	\KwResult{Split $\delta$ at $c$ and update $\A$}
	\Begin{
		Remove $\delta$ from \A\;
		Remove $q'$ and all its descendants from \A\;
		Compute $S^{\delta_1}$ and $S^{\delta_2}$ \;
		$A_1 \gets$ \SPTA{$S^{\delta_1}$} , $A_2 \gets$ \SPTA{$S^{\delta_2}$} \;
		Let $q_1$, $q_2$ be the roots of $A_1$, $A_2$ respectively in:\;
		$\Delta\gets\Delta\cup\{(q,\inter{a}{c},q_1),(q,\inter{c}{b},q_2)\}$\;	
	}
\caption{Splitting a transition: $\mathtt{Split}((\A,R),\delta,c,S)$\label{SplitAlgo}}
\end{algorithm}

The computation of $S^{\delta_i}$ on line 4 can be done in polynomial time by running the automaton on the sample. However, in an actual implementation, these samples can just be stored in the states.

\subsection{Merging}
Although we are constructing guards in a top-down fashion, merging operations still occur: at the beginning, states are merged if and only if they are at the same depth in the (usual sense) PTA, thus other merges have not been considered. 

\begin{algorithm}[h]
	\DontPrintSemicolon
	\SetSideCommentRight
	\SetKwData{True}{true}\SetKwData{False}{false}
	\SetKwFunction{Split}{Split}\SetKwFunction{Merge}{Merge}
	\KwData{An augmented SFA $(\A = (\mathbb{A},Q,q_{\iota},F,\Delta), R)$, two states $q,q' \in Q$ with $q'$ not being red, a sample $S$}
	\KwResult{Update \A\ by merging $q$ and $q'$}
	\Begin{
		Add a new state $q_n$ to \A \;
		\lIf{$q\in F$ or $q'\in F$}{$F \gets F \cup \{q_n\}$}
		\lIf{$q\in R$ or $q'\in R$}{$R \gets R \cup \{q_n\}$}
		\If{there exists a transition $\delta'$ from $q'$}{
			\ForAll{$\delta = (q,\inter{a}{b},\_)$ from $q$, by increasing $a$}{
				\lIf{$b\neq d_\infty$}{\Split{$(\A,R),\delta',b,S$} and $\delta'\gets (q',\inter{b}{d_\infty},\_)$ }
				}
			}
		\ForAll{$\delta = (q_1,\varphi,q_2) \in \Delta$ \tcp*{(replace $q$ and $q'$ by $q_n$ in transitions)}}{ 
			\lIf{$q_1=q$ or $q_1=q'$}{$\delta\gets (q_n,\varphi,q_2)$}
			\lIf{$q_2=q$ or $q_2=q'$}{$\delta\gets (q_1,\varphi,q_n)$}
		}
		\While{$\Delta$ contains two transitions $(q_n,\varphi,q_1)$ and $(q_n,\varphi,q_2)$ with the same guard such that $q_2$ is not red}{
			\Merge{$(\A,R),q_1,q_2,S$}\;
		}
		Remove $q$ and $q'$ from $\A,F,R$\;
		}
\caption{Merging two states: $\mathtt{Merge}((\A,R),q,q',S)$\label{MergeAlgo}}
\end{algorithm}

Merging two states is recursive: successors of the two states (and their successors, etc.) are merged to maintain determinism in the automaton. When this function is called from the main algorithm, it is always with a blue state and a red one: in recursive calls, this might not be the case, but one of the two states will always be guaranteed to not be red. See Figure \ref{MergeFig} for an illustration.

\begin{figure}[t]
	\centering
	\includegraphics[width=0.8\textwidth]{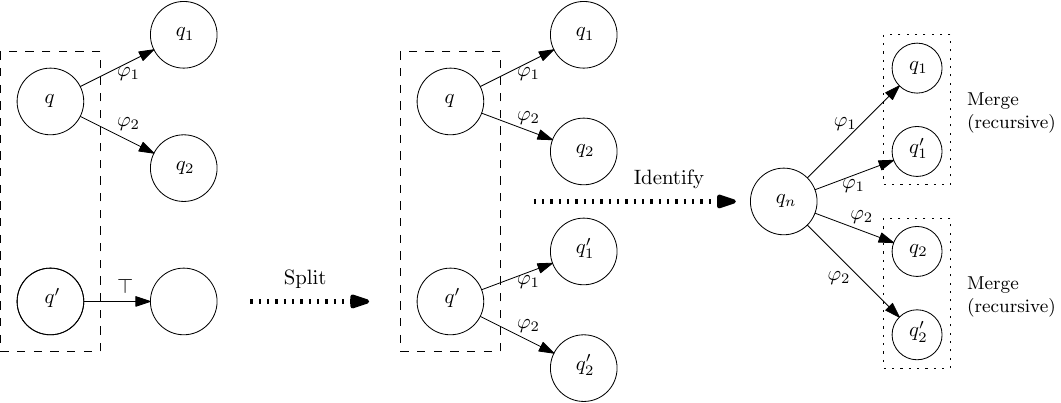}
	\caption{The steps of a merge: splits, identification, then recursive calls}
	\label{MergeFig}
\end{figure}

Algorithm \ref{MergeAlgo} creates a fresh state $q_n$ replacing both $q$ and $q'$ in all appearances in $\Delta$. It is rejecting (resp. accepting) when one of the two is (lines 2-4), and takes the ``highest'' (red>blue>uncolored) color among the two states it replaces. The loop on lines 6 and 7 relies on the fact that $q'$ is not red: when it has an outgoing transition, its guard is $\top$: repeatedly splitting it according to outgoing transitions of $q$ ensures that outgoing transitions from both states are the same (as the outgoing transitions from any state form a partition of \D). Then, the replacement and recursive merges occur on lines 8 through 12.

\paragraph{Checking Inconsistency}\label{CheckInconsistentSect}
During the run of \nomdelalgo, some states are \textit{inconsistent}, i.e. both accepting and rejecting, or two words of the sample reach the same state with the same suffix remaining, but are labeled differently (the former is a special case of the latter with the suffix being the empty word).
This is not a problem and
can be fixed by further splits. However, such a state cannot be colored red, as this would be an \textit{irreversible inconsistency}. Thus, when \nomdelalgo tries the different operations, it tests if the obtained automaton can still be made consistent with the sample without modifying red states. If not, the operation needs to be undone (see Algorithm~\ref{FullAlgo}).

\begin{algorithm}[h]
	\DontPrintSemicolon
	\SetKwData{True}{true}\SetKwData{False}{false}
	\SetKwFunction{Split}{Split}\SetKwFunction{Merge}{Merge}\SetKwFunction{SPTA}{SPTA}
	\KwData{An augmented SFA $(\A = (\mathbb{A},Q,q_{\iota},F,\Delta), R)$, a sample $S$}
	\KwResult{\True iff $(\A,R)$ can still be made consistent with $S$}
	\Begin{
		
		\ForAll{red states $q \in Q$}{
			\lIf{$q\in F$ and $q\in R$}{\Return{\False}}
			\ForAll{$(q,\varphi,q') \in \Delta$ such that $q'$ is not red}{
			Let $S^\delta = (S^\delta_+,S^\delta_-)$ be the set of tails of $S$ for $\delta$\;
			\lIf{$S^\delta_+\cap S^\delta_-\neq\emptyset$}{\Return{\False}}
			}
		}
		\Return{\True}
		}

\caption{\label{CheckInconsistentAlgo}\begin{small} Checking for irreversible inconsistency: $\mathtt{Consistent}(\A,R,S)$ \end{small}}
\end{algorithm}

\subsection{The \nomdelalgo Algorithm}

Algorithm \ref{FullAlgo} below contains 
our passive learning algorithm for symbolic automata
\nomdelalgo.
Let us define as the \textit{shortlex blue state}, and denote as $q_b$, the blue state reached by the smallest string in shortlex order (shorter strings before longer strings, and strings of the same length ordered by lexicographic order). At each step, \nomdelalgo first tries to merge it into each red state, then to color it, and if it can't do either, splits it. See Figure \ref{exemple} for an example run of the algorithm.

\begin{algorithm}[h]
	\DontPrintSemicolon
	\SetKwData{True}{true}\SetKwData{False}{false}
	\SetKwFunction{Split}{Split}\SetKwFunction{Merge}{Merge}\SetKwFunction{SPTA}{SPTA}\SetKwFunction{Consistent}{Consistent}\SetKwFunction{Partition}{Partition}
	\KwData{A sample $S$}
	\KwResult{An SFA consistent with $S$}
	\Begin{
		$(\A = (\mathbb{A},Q,q_{\iota},F,\Delta),R)\gets$ \SPTA{$S$}\;
		Color the initial state $q_\iota$ of \A\ red\;
		
		\While{\A\ contains non-red states}{
				\lForAll{$(q,\varphi,q') \in \Delta$ such that $q$ is red and $q'$ is not}{Color $q'$ blue}
				Let $q_b$ be the shortlex blue state \;
				\ForAll{red states $q \in Q$}{
					\Merge{$(\A,R),q,q_b,S$}\;
					\leIf{\Consistent{$\A,R,S$}}{skip}{Undo the merge}
				}
				Color $q_b$ red\;
				\leIf{\Consistent{$\A,R,S$}}{skip}{Undo the coloring}

				Let $\delta$, $S^\delta$ be the transition leading to $q_b$ and its set of tails \;
				\ForAll{tails $t\cdot\tau \in S^\delta$ (picked by increasing $t$)}{
					\Split{$(\A,R),\delta,t,S$}\;
					\eIf{\Consistent{$\A \text{ with the new shotlex state colored red},R,S$}}
					{
						Undo the split\;
					}{
						Undo the split and Redo the previous split tried\;
						skip
					}
				}
		}
		\Return{\A}
		}
\caption{The \nomdelalgo algorithm\label{FullAlgo}}
\end{algorithm}

By skip, we mean go to the next iteration of the line 4 while loop.
On line 3, as we assume the sample does not contain a word both positively and negatively, the initial state (corresponding to $\varepsilon$) cannot be both accepting and rejecting, so it is colored right away.
In lines 13 and onward, our method of finding the right split is that we are trying to find the split with the largest value such that the resulting new shortlex blue state can be colored red. Thus, as long as splitting and coloring the new $q_b$ red doesn't create an inconsistency, we keep trying bigger and bigger splits. The first one that fails tells us that the previous was the right one (line 18).

In the example run of Figure \ref{exemple}, only the action that is actually performed is shown at each step. For example at the very first step, the only blue state is both accepting and rejecting, so coloring it red, or merging it into the red state, would create an irreversible inconsistency: thus, the only possible action is to split the transition leading to it. Splitting at value 0 works, but splitting at 100 does also, while any larger value would not: thus, the split at 100 is performed.

In the coming section \ref{propsect}, we show that \nomdelalgo runs in polynomial time and returns a SFA consistent with its input sample and we prove the main result of this paper, stating that there exist polynomial characteristic samples for \nomdelalgo, computable in polynomial time. That is, \emph{\nomdelalgo efficiently identifies in the limit the class of SFAs over monotonic algebras.}
\begin{figure}[!h]
\centering
	\includegraphics[width=0.8\textwidth]{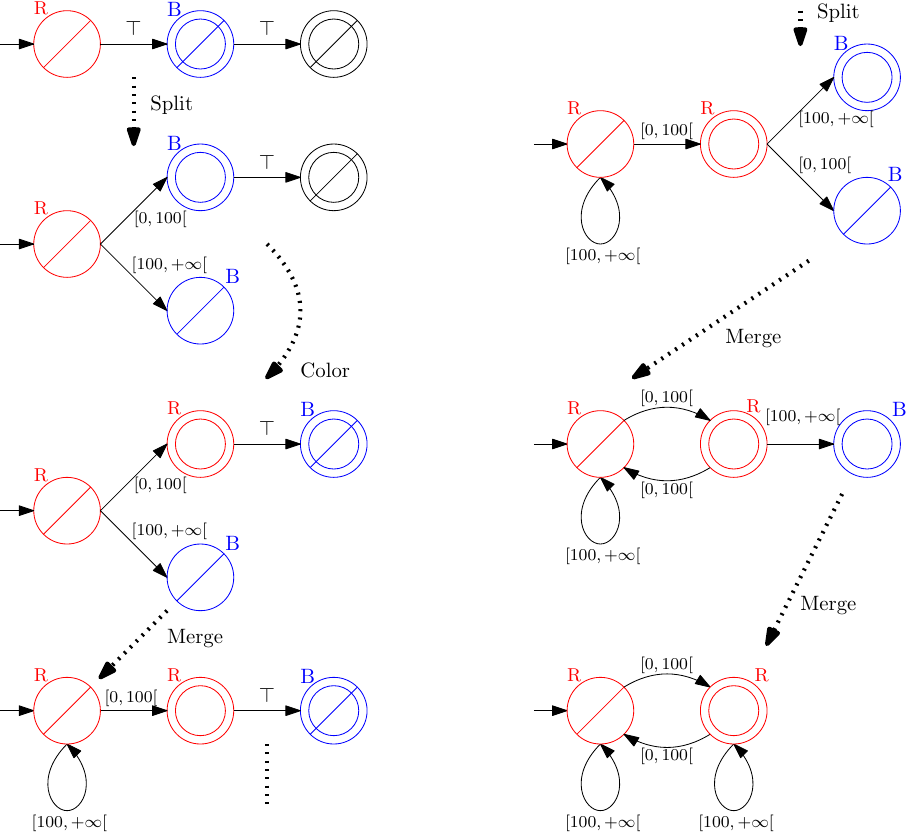}
    
	\caption{An example run of \nomdelalgo for the sample: $\{ (\varepsilon,-), (0,+), (100,-), (0\cdot 0,-),$ $(0\cdot~100,+) \}$ over the interval algebra. Capital letters next to states denote their color (Blue, Red) when they have one.}
	\label{exemple}
\end{figure}

Let us detail the characteristic sample here. Let $\L_t$ be any SFA language, and $\A_t$ the minimal-state neat SFA recognizing it (it exists and is unique \cite[Lemma 5.7]{FFZ}).
Given a state $q \in Q_t$, denote its shortlex access word by $\tau_q$. Given $q,q'$ two distinct states, denote as $\delta_{q,q'}$ a distinguishing word between the two (that is, a word accepted from $q$ iff it is rejected from $q'$). These words' existence is guaranteed by the minimality of $\A_t$; they are all of length linear in $|Q_t|$, and can be computed in polynomial time.
The characteristic sample is then defined as: 
	\[\bigcup_{q_1 \in Q_t}\bigcup_{\substack{a \in \D \text{ s.t.}\\ \exists(q_1,\inter{a}{b},q_2) \in \Delta_t}} \bigcup_{q \neq q_2} \tau_{q_1}\cdot a\cdot\delta_{q_2,q} \cup \bigcup_{q,q' \in Q_t^2} \tau_q\delta_{q,q'} \]

		\section{Properties of \nomdelalgo}\label{propsect}
	 	
In this section, we first show some basic properties of the \nomdelalgo algorithm, namely termination (in polynomial time) and correction. Then, we prove that \nomdelalgo admits characteristic samples for all SFA languages (over monotonic algebras). The fact that these sets can be computed in polynomial time, together with lemma \ref{polytime}, means that \nomdelalgo learns SFA over monotonic algebras \textit{in the limit}.

\begin{lemma}\label{polytime}
	\nomdelalgo terminates in time polynomial in the size $\tS = \sum_{\tau \in S} \mid\!\tau\!\mid$ of the input sample $S$.
\end{lemma}
\begin{proof}
	First, the construction of the SPTA is clearly polynomial in the size of the sample.

	Then, each operation takes polynomial time:

	For splitting, the sets of suffixes can be constructed in polynomial time by running $\A$ on the relevant part of the sample. The two SPTA constructions then take time polynomial in the size of these.

	For merging, some splits might be necessary in the outgoing transitions from the blue state, but as each one is due to a state reached by a word of the sample, the number of splits is bounded by the size of the input sample.

	Promoting a state takes constant time. 
	
	Thus, we need to show that a polynomial number of iterations takes place, and that in each of these, the choice of an operation takes polynomial time.

	Each split tried is due to a tail of the sample: there are $\tS$ of those (one for each position of each word). Similarly, each merge takes place with a red state reached by some word, thus there are only $\tS$ to try.

	When a state has been promoted, two states have merged, or a split has separated two words, this operation will not take place again.
	Thus, the number of iterations of the algorithm is bounded by $2\!\tS$, and at each one, at most $2\!\tS$ operations will be tried: this concludes the proof.
\end{proof}
\begin{lemma}
	Given an input sample $\{S_+,S_-\}$, \nomdelalgo returns a neat SFA $\A$ consistent with the input sample: that is, $S_+\subset\L(\A)$ and $S_-\subset\L(\A)^{\mathsf{c}}$
\end{lemma}
\begin{proof}
	First, $S_+\cap S_-=\emptyset$, so the coloring of the initial state does not make $\A$ inconsistent. Then, all coloring and merging are only done if they create no irreversible inconsistency ; thus, during the run, $\A$ is never irreversibly inconsistent.
	
	When the algorithm terminates, all states are colored red. At this point, there is no state that is both accepting and rejecting (that would be a irreversible inconsistency), and by construction of the SPTA, every positive sample ends in an accepting state and every negative sample ends in a rejecting state.
	As the algorithm does terminates by the previous lemma, this concludes the proof.
\end{proof}
\begin{theorem}
	There exists polynomial characteristic samples for \nomdelalgo, computable in polynomial time.
\end{theorem}

\begin{proof}
	Let $\L_t$ be any SFA language, and $\A_t$ the minimal-state neat SFA recognizing it (it exists and is unique by \cite[lemma 5.7]{FFZ}).

	Notice that the fixed order between operations makes it so that, at every iteration, only the transition leading to $q_b$ with the smallest (in terms of $\prec$) guard is considered, and words in the sample not taking this transition will have no impact. With this in mind, we show that a few words allow us to identify every transition in $\A_t$ ; the union of these over all transitions is thus a characteristic sample for $\L_t$.

	When we say that a given word is in the sample, it is implied that it is added with its corresponding label (depending on whether it is in $\L_t$ or not).

	Given a state $q \in Q_t$, denote $\tau_q$ its shortlex access word. Given $q,q'$ two distinct states, denote as $\delta_{q,q'}$ a distinguishing word between the two (that is, a word accepted from $q$ iff it is rejected from $q'$). These words' existence is guaranteed by the minimality of $\A_t$ ; they are all of length linear in $\mid\!Q_t\!\mid$, and can be computed in polynomial time.

	Our sample is defined as: 
	\[\bigcup_{q_1 \in Q_t}\bigcup_{\substack{a \in \D \text{ s.t.}\\ \exists(q_1,\inter{a}{b},q_2) \in \Delta_t}} \bigcup_{q \neq q_2} \tau_{q_1}\cdot a\cdot\delta_{q_2,q} \cup \bigcup_{q,q' \in Q_t^2} \tau_q\delta_{q,q'} \]

	Note that for any transition out of a fixed state $q_r$ with guard $\inter{a}{b}$ and target $q$, there is also an outgoing transition with guard $\inter{b}{c}$ going to another state $q' \neq q$ (it exists because $\A_t$ is complete, and goes to a different state because it is minimal). Thus, $\tau_{q_r}\cdot a \cdot\delta_{q,q'}$ and $\tau_{q_r}\cdot b \cdot\delta_{q,q'}$ are both present in the sample, \emph{with one being accepted and one being rejected}.

	Thus, at a general step in the algorithm, the situation is as in Figure \ref{casdebase}:
	
	\begin{figure}[h]
		\centering
		\includegraphics[width=0.6\textwidth]{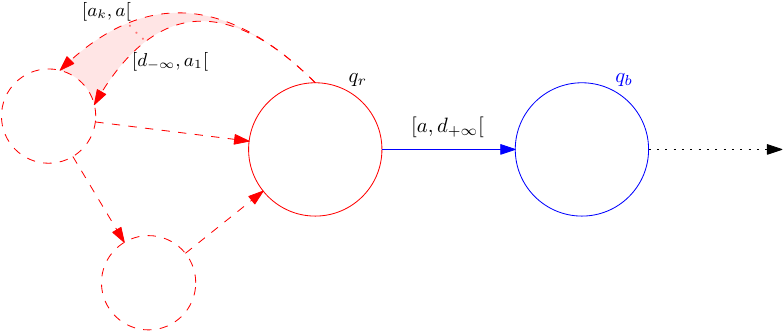}
		\caption{$q_b$ is the shortlex blue state, but is incoherent and needs to be split before being colored or merged. Dashed transitions and states are figurative.}
		\label{casdebase}
	\end{figure}
	
	The outgoing transitions from red state $q_r$ with guards up to $a$ have been identified, and the next transition to be identified out of $q_r$ is $(q_r,\inter{a}{b},q)$. Note that $a$ was identified when learning previous transitions\footnote{When $a = d_{-\infty}$, then the transition will be the first identified out of $q_r$: thus the lower bound will, correctly, be $d_{-\infty}$ and the rest of the reasoning is the same.}: the only things to be learned are the upper bound $b$ and the target $q$.
	At this point, $q_b$ cannot be merged with a red state nor colored blue, because the two aforementioned words would create an irreversible inconsistency: the transition leading to it should be split. Any split with value $b' \succ b$ would yield the same issue on the newly created $q_b$. Splits with value $b' \prec b$ would create no inconsistency, but the one with value $b$, as the last consistent one, is the one executed, identifying the upper bound of the transition. 
	\begin{figure}[h]
		\centering
		\includegraphics[width=0.6\textwidth]{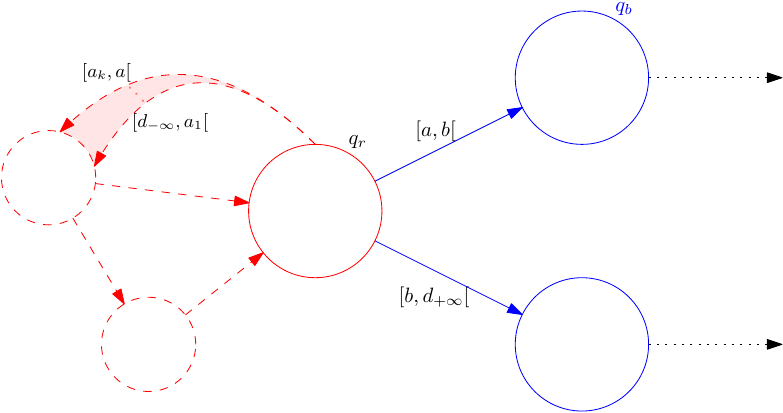}
		\caption{The situation after the split: the new shortlex blue state can be merged or colored consistently.}
		\label{casdebase2}
	\end{figure}

	Once this split is done (Fig. \ref{casdebase2}), the next step will identify the target of the transition. The newly created $q_b$ is not inconsistent, and merges with all known red states will be tried. However, due to the presence of $\tau_{q_r}\cdot a \cdot\delta_{q,q'}$ and $\tau_{q'}\delta_{q,q'}$ for all states $q' \neq q$, always with opposite labels (as, in $\A_t$, $\tau_{q_r}\cdot a$ reaches $q$ and $\tau_{q'}$ reaches $q'$), the only possible merge is with $q$, and if $q$ has not been colored red already, then no merge is possible and $q_b$ is colored red (thus becoming $q$ for all intents and purposes when identifying further transitions).

	Thus, we have proved that every transition of $\A_t$ will be correctly identified with this sample.

	As long as not all states are accepting or not all states are rejecting (in which case the minimal automaton has only one state and a charcteristic sample containing only $\varepsilon$), for each state $q$ there exists $q'$ such that $\delta_{q,q'} = \varepsilon$. This means that all access words are present in the sample, and ensures that all states of $\A_t$ are correctly identified as final if and only if they should be.

	The fixed order over operations, independently of the sample, also shows that even if more examples were added, this would not make \nomdelalgo return another automaton (assuming all words are correctly labeled according to $\L_t$), which concludes the proof. 
	
\end{proof}

\textit{Remark :} Note that the words $\tau_{q_1}\cdot a\cdot\delta_{q_2,q}$ appearing for all transitions with lower bound $a$ make it so that when a transition has guard $\inter{d_{-\infty}}{b}$, there will be words in the sample with $d_{-\infty}$ as a letter. When dealing with the interval algebra over \N\ , $d_{-\infty} = 0$ so this is not an issue. When dealing with \R\, this means some words will contain $-\infty$ as a letter, which might be inelegant or unpractical: however, finding the minimal lower bound $m$ of any guard and replacing every instance of $-\infty$ by $m-1$ (or similar) will allow us to do without it. But when dealing with a general monotonic algebra with no known structure, $d_{-\infty}$ will appear as a concrete letter.

		\label{Conclusionsect}
	 	\section{Comparison to Previous-Known Results}\label{comparison}
	We compare our results to Fisman \etal\ \cite{FFZ}, the, to the best of our knowledge, only previously known algorithm for passive learning of SFA. 
	With the same formalism as for \nomdelalgo's the characteristic sample for Fisman \etal's algorithm is:
	\[\bigcup_{\substack{a \in \D \text{ s.t.}\\ \exists(q_1,\inter{a}{b},q_2) \in \Delta_t}} \bigcup_{\substack{q \in Q_t \\ q' \neq q'' \in Q_t}} \tau_{q}\cdot a\cdot\delta_{q',q''} \cup \bigcup_{q,q' \in Q_t^2} \tau_q\delta_{q,q'} \]
	Intuitively, if $a$ appears as the lower bound of some transition, then there are words indicating at each state to which state reading $a$ leads (recall $\tau_{q}$ is the access word to $q$, and $\delta_{q',q''}$ is a word distinguishing $q'$ from $q''$); while we only ask that a letter appears after the access word for a state if it does appear as a lower bound of a transition from that state.
	In the worst case (if all lower bounds of transitions are distinct), Fisman \etal's CS contains $\left|\Delta\right|\left(\left|Q\right|-1\right)\left|Q\right|^2 + \left|Q\right|^2$ words. The one of \nomdelalgo instead has always $\left|\Delta\right|\left(\left|Q\right|-1\right) + \left|Q\right|^2$ words. In both, the words' lengths are linear in $\left|Q\right|$.
	Furthermore, the CS for \nomdelalgo is included in the one for Fisman \etal's algorithm. Therefore, if the latter is guaranteed to learn an automaton, \nomdelalgo also learns it, whereas the converse is false.

	For example, consider the sample of Figure \ref{exemple} with
        one $100$ replaced by $99$: $\{ (\varepsilon,-), (0,+), (100,-),(0\cdot 0,-),(0\cdot 99,+) \}$.
        %for learning an automaton over the interval algebra.
        \nomdelalgo handles it like before, the only difference being that outgoing transitions from the right state are guarded by $[0,99[$ and $[99,+\infty[$. Yet, in Fisman \etal's algorithm, that sole $99$ is removed by their \emph{decontaminating} function, so the DFA learner uses sample $\{ (\varepsilon,-), (0,+)$, $(100,-)$, $(0\cdot 0, -)\}$ and learns an automaton accepting $0$ and rejecting everything else. After being lifted back to symbolic automata, the automaton is not consistent with the sample, and thus, the algorithm returns the prefix-tree automaton for the sample 
		accepting only two words, thus not generalizing the sample.

	We also compare our results to the ones from Verwer \etal\ \cite{VWW}: the \emph{real-time automata} (RTA) their RTI algorithm learns are similar to SFA. They recognize words that have letters from a finite alphabet labeled with natural numbers corresponding to delays between letters. RTA over a single-letter alphabet are thus SFA over the natural interval algebra. \nomdelalgo thus generalizes RTI by using any monotonic algebra and not only \N: it could be adapted to recognize words over a finite alphabet with delays from any monotonic algebra by redefining SPTAs for this case. Also, RTI relies on the structure of \N\ by asking, when learning a transition with lower bound $n$, for both $n$ and $n-1$ to appear in the sample. The final main difference is that characteristic samples are not defined explicitly but by giving an algorithm to generate them: running RTI with an empty sample, then each time a transition not in the target automaton is learned, adding words to the sample that forbid this operation, and restarting. This makes the CS smaller than the ones we defined ---although \nomdelalgo would also learn the target automaton given RTI's CS---, but practically impossible to express explicitly.

\section{Conclusion and Future Work}\label{futurework}
	A way to generalize \nomdelalgo is to handle the non-monotonic case. As it is known \cite{FFZ} that SFAs over the propositional algebra cannot be efficiently identifiable unless $P=NP$, this can take various forms:
	\begin{enumerate}[nolistsep]
        \item defining a polynomial-time algorithm without polynomial-size characteristic samples,
		\item defining an algorithm that doesn't run in polynomial-time (with or without polynomial-size characteristic samples),
		\item finding sufficient and/or necessary conditions for an algebra for corresponding SFA to be efficiently identifiable.
	\end{enumerate}
	We expect the third option to be the most promising one, as results
        for active learning are similar: Argyros and D'Antoni's $\mathit{MAT}^*$ \cite{MATstar} showed that predicates of a boolean algebra being learnable in polynomial time using membership (of letters) and equivalence queries, is a sufficient condition for SFA over this algebra to be learnable in polynomial time using membership (of words) and equivalence queries by creating instances of the predicate learning algorithms. Similarly, we would like to define a framework that relies on passive predicate-learning algorithms to get automata-learning ones.

	Another generalization would be using the \textit{Evidence-Driven State Merging} (EDSM) framework, where operations do not happen in a fixed order, but a heuristic, called \textit{evidence value}, leads us to do ``the most productive'' one at every iteration  (see \cite{EDSMintro} for a presentation, and \cite{EDSMsurvey} for a survey of evidence values). Although it was defined in cases where the only operations were coloring and merges, results with splits are promising \cite{VWW}. \nomdelalgo is currently equivalent to an EDSM algorithm with evidence value such that consistent merges score higher than the coloring, who scores higher than splitting (all only of the shortlex blue state, any other operation scoring 0). The existence of characteristic samples for these evidence-based algorithms is not always guaranteed, but in practice they converge quickly towards small automata consistent with the sample.

	\subsubsection*{Acknowledgments}
	This work was partially supported by the ``France 2030'' government investment plan managed by ANR, under the reference ANR-23-PEIA-0006.

\bibliographystyle{splncs04} % LNCS: style
\bibliography{biblio}
\end{document}